\documentclass[journal,letterpaper,twocolumn,twoside,nofonttune]{IEEEtran}
\usepackage{multirow}

\usepackage{algorithm, algorithmic}
\usepackage[utf8]{inputenc} 
\usepackage[T1]{fontenc}
\usepackage{url}
\usepackage{ifthen}
\usepackage{cite}
\usepackage{amsmath}%
\usepackage{wasysym}%
\usepackage{amssymb}
\usepackage{enumitem}
\usepackage{mdwmath}
\usepackage{blindtext}
\usepackage{eqparbox}
\usepackage{stfloats}

\usepackage{verbatim}
\usepackage[english]{babel}
\usepackage{lipsum}
\usepackage{xcolor}
\usepackage{tikz}

\usepackage{subfigure}
\usetikzlibrary{matrix,shapes,arrows,positioning,chains, calc}
\usepackage{graphicx}
\usepackage{algorithm, algorithmic}

\title{\Huge$\,$\\[-2.75ex]
{Physical Layer Secret Key Generation\\ in Static Environments}\\[0.50ex]}
\author{\large%
Nasser Aldaghri,\,\,\IEEEmembership{Student Member,~IEEE}, and
Hessam Mahdavifar,\,\,\IEEEmembership{Member,~IEEE}\\

\thanks{%
The material in this paper was presented in part at the IEEE Global Communications Conference in December 2018 \cite{aldaghri2018fast}. This work was supported in part by the National Science Foundation under grants CCF--1763348, and CCF--1909771.} 
\thanks{N.\, Aldaghri and H.\ Mahdavifar are with the Department of Electrical Engineering and Computer Science, University of Michigan, Ann Arbor, MI 48109 (email: aldaghri@umich.edu and hessam@umich.edu).}
}

\newtheorem{theorem}{{Theorem}}
\newtheorem{lemma}[theorem]{{Lemma}}

\newtheorem{definition}{{Definition}}

\DeclareMathAlphabet{\mathbfsl}{OT1}{ppl}{b}{it} 

\newcommand{\be}[1]{\begin{equation}\label{#1}}
	\newcommand{\ee}{\end{equation}} 
\newcommand{\eq}[1]{(\ref{#1})}

\renewcommand{\leq}{\leqslant}
 
\renewcommand{\geq}{\geqslant}

\newcommand{\Tref}[1]{Theo\-rem\,\ref{#1}}

\newcommand{\Lref}[1]{Lem\-ma\,\ref{#1}}
\newcommand{\Cref}[1]{Co\-ro\-lla\-ry\,\ref{#1}}

\newcommand{\deff}{\mbox{$\stackrel{\rm def}{=}$}}

\begin{document}

\maketitle
\begin{abstract}
Two legitimate parties, referred to as Alice and Bob, wish to generate secret keys from the wireless channel in the presence of an eavesdropper, referred to as Eve, in order to use such keys for encryption and decryption. In general, the secret key rate highly depends on the coherence time of the channel. In particular, a straightforward method of generating secret keys in static environments results in ultra-low rates. In order to resolve this problem, we introduce a low-complexity method called \textit{induced randomness}. In this method, Alice and Bob independently generate local randomness to be used together with the uniqueness of the wireless channel coefficients in order to enable high-rate secret key generation. In this work, two scenarios are considered: first, when Alice and Bob share a direct communication channel, and second, when Alice and Bob do not have a direct link and communicate through an untrusted relay. After exchanging the induced randomness, post-processing is done by Alice and Bob to generate highly-correlated samples that are used for the key generation. Such samples are then converted into bits, disparities between the sequences generated by Alice and Bob are mitigated, and the resulting sequences are then hashed to compensate for the information leakage to the eavesdropper and to allow consistency checking of the generated key bit sequences. We utilize semantic security measures and information-theoretic inequalities to upper bound the probability of successful eavesdropping attack in terms of the mutual information measures that can be numerically computed. Given certain reasonable system parameters this bound is numerically evaluated to be $2^{-31}$ and $2^{-10.57}$ in the first and the second scenario, respectively.
\end{abstract}

\section{Introduction}\label{introduction}
Wireless networks are becoming increasingly distributed in future systems, e.g., the fifth generation of wireless networks (5G) and the Internet of Things (IoT), which, consequently, poses a higher risk of malicious attacks against message confidentiality in these systems. In general, communication devices secure messages using either symmetric-key encryption schemes such as Advanced Encryption Standard (AES) \cite{AES}, or asymmetric-key encryption schemes such as Rivest–Shamir–Adleman (RSA) \cite{RSA}. Asymmetric-key schemes are not preferred for devices with limited resources, e.g., as in IoT networks, due to their complex mathematical operations. Instead, symmetric-key cryptographic schemes are desired in IoT networks due to their low-complexity implementations \cite{granjal2015security}. Such schemes require the secret keys for the encryption and decryption to be distributed beforehand between the legitimate parties. To complement the symmetric-key cryptographic schemes, physical layer security methods can be deployed to exchange secret keys between the nodes in order to be used in the encryption and the decryption algorithms \cite{bloch2011physical}. 

The fundamental works of \cite{ahlswede1993common, maurer1993secret} established an information-theoretic framework to study the use of common randomness for secret key generation. In practice, characteristics of wireless links are shown to provide a great source for the common randomness to be used for secret key generation, which have recently received significant attention \cite{mathur2008radio, zhang2016key}. More specifically, the wireless channel has two main features that are essential for secret key generation, namely, reciprocity and randomness. The wireless channel is reciprocal over each single coherence time interval \cite{wilson2007channel}, and it has inherent randomness due to the variation of the channel coefficients between different coherence time slots \cite{mathur2008radio}. Note that the former requires an underlying synchronization mechanism while the latter assumes a dynamic environment. These features are often assumed to be available to the wireless nodes, i.e., the legitimate parties, which can then be utilized in low-complexity secret key generation protocols at the physical layer. The setup for the key generation protocols is as follows: the legitimate parties Alice and Bob share a common wireless channel, either directly or indirectly through a relay node. They communicate through this channel with the goal of generating a common secret key bit sequence, while keeping a passive eavesdropper Eve oblivious about the generated key. Such protocols often include the following steps \cite{zhang2016key}:
\begin{enumerate}
    \item Randomness sharing: In this step, the legitimate parties observe correlated samples from a common source of randomness, e.g., wireless channel coefficients.
    \item Quantization: This is the process of converting such correlated samples, which are often real-valued, into binary bits.
    \item Reconciliation: In general, there is a mismatch between the binary sequences observed and quantized by Alice and Bob. Reconciliation is the process of mitigating such mismatch between Alice's and Bob's bit sequences using methods such as cosets of binary linear codes.
    \item Privacy amplification: This is the process of compensating for the information leakage to the eavesdropper Eve during the aforementioned steps.
\end{enumerate}

\subsection{Related Work} \label{relatedwork}
This section provides an overview of related work on secret key generation using characteristics of wireless channel under two main scenarios; the first scenario where the legitimate parties Alice and Bob have a direct communication channel as the only means of communication, and the second scenario where their communication is helped by a relay node. 
    
\subsubsection{Secret key generation over direct communication channels} \label{dskgintro}
In this case, different characteristics of the wireless channel can be utilized as the source of common randomness in secret key generation protocols. This includes the channel state information (CSI), the received signal strength (RSS), and the channel phase, just to name a few \cite{ren2011secret}. As mentioned earlier, there are two main underlying assumptions in such protocols. First, the assumption on reciprocity of the wireless channel guarantees the reliability of such protocols. Second, the randomness of the key is guaranteed by the assumptions on temporal decorrelation \cite{molisch2012wireless}. The resulting secret key generation protocols, e.g., \cite{asbg,mathur2008radio,wilson2007channel,cgc,keep,liu2012exploiting, zhang2014secure}, where orthogonal frequency division multiplexing (OFDM) is utilized to increase the key rate in \cite{zhang2014secure}, often require dynamic environments in order to satisfy the second assumption and to enable secret key generation at non-zero rates. It is worth noting that some imperfections of the channel measurements may occur due to mismatched hardware and synchronization errors \cite{zeng2015physical}.

Wireless channels can be naturally assumed to be dynamic assuming a certain level of mobility by users and/or in the surrounding environment. However, such assumptions do not hold in static environments such as indoor IoT networks. Consequently, the aforementioned protocols result in ultra-low/zero secret key rates in such environments. This issue has been studied in the literature and various solutions have been proposed. Solutions include utilizing multiple-input-multiple-output (MIMO) antennas systems \cite{zorgui2016ergodic,jorswieck2013secret,jiao2018secret}, beamforming \cite{madiseh2012applying}, deploying friendly jamming \cite{gollakota2011physical} where the users act as jammers to confuse the eavesdropper, and using artificial noise to confuse the eavesdropper \cite{goel2008guaranteeing}. In another line of work, some user-introduced randomness is utilized for various purposes \cite{wang2011fast,huang2013fast,zeng2015physical,li2017security,fang2017manipulatable}. For instance, the user's randomness is used to counter certain types of attacks by the eavesdropper in \cite{zeng2015physical}. However, the use of induced randomness for key generation, and more specifically to increase the key rate in static scenarios, is not discussed in \cite{zeng2015physical}. In general, prior schemes that use some user-introduced randomness require complex underlying architectures, e.g., MIMO transceivers, or unconstrained sources of randomness, i.e., continuous sources whose Shannon entropy is infinity, which are expensive to implement \cite{sunar2007provably}. This, in turn, makes them unappealing for applications where nodes experience a static environment and have limited resources, e.g., IoT networks, sensor networks, etc. Also, solutions based on utilizing coupled dynamics existing in synchronization mechanisms \cite{coupling1,coupling2} and based on full-duplex communications \cite{IMS} are proposed for low-complexity IoT networks, which are often limited to very-short-range communications due to power constraints for implementing coupled dynamics in practice. Moreover, several prior works have considered utilizing relays for secret key generation, as discussed next.

\subsubsection{Secret key generation with the help of a relay} \label{rskgintro}
In this case, there exists a relay node that assists Alice and Bob to generate the shared secret keys. The wireless characteristics used for the randomness sharing in the first scenario, e.g., CSI and RSSI, can be similarly applicable here. Various methods have been proposed in the literature to utilize relays in order to improve the key generation rate when Alice and Bob have a direct communication link as well \cite{lai2012cooperative,wang2012cooperative}. The use of relays in generating secret keys when Alice and Bob do not have a direct communication link is studied in \cite{dong2010improving,shimizu2011physical,zhou2014secret}. A major arguable assumption in these related works is that the relay nodes are trusted. However, the wireless nodes, especially when they are considered low-complex and low-cost as in IoT networks, are susceptible to hacking, even after the key generation process is done. Hence, it is highly desirable to ensure that limited information about the generated secret key is leaked to the relay throughout the process. This is the motivation behind several other related works which assumed the relay nodes are untrusted. For instance, a method to accommodate this case by utilizing friendly jamming is introduced in \cite{zhang2012physical}. Another method that requires a moving relay to generate secret bits is proposed in \cite{guillaume2015secret}. Also, a novel method to resolve the issue of untrusted relays using a MIMO architecture is suggested in \cite{thai2016untrusted}. Such methods, however, require dynamic environments. As mentioned before, these protocols are not appealing for applications where nodes are resource-constrained and the environment is static.
    
\subsection{Our Contributions}\label{contintro}
Our main contribution in this work is a solution, based on low-complexity methods, for resolving the issue of low/zero rate secret key generation between two legitimate nodes in static environments. In the proposed solution, we utilize induced randomness generated by the legitimate parties and exchanged between them. More specifically, Alice and Bob independently generate a certain number of random bits. Then, they map these bits to quadrature amplitude modulation (QAM) symbols which they exchange using the direct communication channel (the first considered scenario) or through an untrusted relay (the second considered scenario). After the exchange of the generated randomness, Alice and Bob process their received sequences, which are the generated randomness by the other party and passed through the channel, using their own random sequences. The reciprocity of the channel/channels ensures that they obtain highly correlated sequences. Such common noisy randomness is then used to extract shared secret keys by following quantization, reconciliation, and privacy amplification steps. The reliability and the security of the proposed protocols are analyzed by upper bounding the probability of falsely accepting a mismatched secret key and the probability of a successful eavesdropping attack by Eve, respectively. While most of prior works on designing physical layer secret key generation protocols rely on spatial decorrelation assumptions to guarantees the security of the key, we provide, to the best of our knowledge, the first rigorous result on upper bounding the probability of successful eavesdropping attack in such protocols. It is worth noting that although the motivation behind the design of the proposed protocols is to resolve the issue of environment immobility, they work in dynamic environments as well assuming that the wireless channel does not change during each session of randomness exchanges.

The proposed protocols are considered under two major scenarios. In the first scenario, secret key generation over a direct communication channel is considered, which was presented in part in \cite{aldaghri2018fast}. In the second scenario, secret key generation with the help of an untrusted relay is considered assuming that there is no direct communication link between Alice and Bob. In the proposed protocols, a communication scheme based on OFDM is assumed to increase the secret key rate as in \cite{zhang2014secure}. Furthermore, we utilize secure sketch \cite{juels1999fuzzy} and universal hash functions (UHF) \cite{hash} to ensure reliability and security of the generated keys while enhancing the randomness of the key bit sequences. Numerical results are provided for the proposed protocol assuming reasonable parameters in the communication setup. These parameters include the modulation order, the number of OFDM subcarriers, the signal-to-noise ratio (SNR), and the quantization resolution. Then, various fundamental metrics are characterized including the bit generation rate (BGR), the bit mismatch rate (BMR), the bit error rate (BER), and the randomness of the key generated using the National Institute of Science and Technology (NIST) randomness tests \cite{rukhin2001statistical}. In addition, a setup in which realistic channel coefficients for 5G millimeter wave (mmWave) channels are generated by the NYUSIM Channel Simulator \cite{NYUSIM} is considered, assuming the first scenario, in order to evaluate the protocol in a realistic environment. Furthermore, we introduce a new efficiency measure for protocols that utilize induced randomness. This parameter, called \textit{randomness efficiency}, measures what percentage of the induced randomness is utilized in the generated common random sequence. The randomness efficiency in the first scenario is 50\,\%, while it is 33\,\% in the second scenario. 
    
The rest of this paper is organized as follows. The system models for the two considered scenarios are discussed in Section \ref{system_model}. In Section \ref{pp}, the proposed protocols for generating secret keys are discussed. In Section \ref{attacker}, the security of the proposed protocols is analyzed. Numerical results are provided in Section \ref{numerical_results}. Finally, the paper is concluded in Section \ref{conclusion}. 
    
\section{System Model}\label{system_model}
Secret key generation protocols consist of two legitimate parties Alice and Bob who aim to generate a common, random, and secure bit sequence using an authenticated shared wireless channel between them. In addition to Alice and Bob, there is an authenticated relay node named Carol, who is honest but curious, and is able to help Alice and Bob generate such keys by relaying their signals when no direct channel exists between them. As the legitimate parties execute the secret key generation protocol, a passive eavesdropper Eve is observing all communications between Alice, Bob, and Carol, and tries to learn as much information as possible about the secret key being generated and shared between Alice and Bob.

\subsection{Direct Secret Key Generation}\label{sys_direct}
\begin{figure}
    \begin{center}
    \includegraphics[trim=112.26pt 353.61pt 154.54pt 258.39pt, clip,width=0.95\columnwidth]{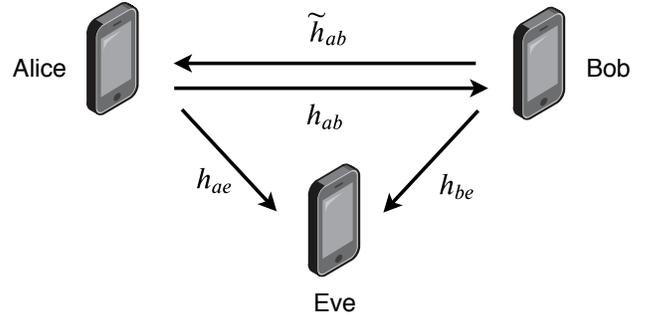}
    \caption{System model for direct secret key generation.}
    \label{fig:DSKG_CH_Model}
    \end{center}
\end{figure}
The channel between Alice and Bob is assumed to be an authenticated wireless channel, but it is not secure. The eavesdropper Eve is assumed to be a passive eavesdropper. The setup of the considered secret key generation (SKG) system is shown in Figure\,\ref{fig:DSKG_CH_Model}. The wireless channel considered in this work is assumed to be a fading channel. Suppose that Alice transmits a signal $ x_{\mathrm{Alice}}(t) $ to Bob, he receives
\begin{align}\label{bobreceivedd}
    y_{\mathrm{Bob}}(t)=x_{\mathrm{Alice}}(t) \circledast h_{ab} (t) + n_b(t),
\end{align}
where $t$ denotes the time, $\circledast$ denotes the convolution operator, $h_{ab} (t)$ denotes the circularly-symmetric Gaussian-distributed channel response with mean $ 0 $ and variance $ \sigma_h^2/2$ in each dimension, and $n_b(t)$ denotes the circularly-symmetric Gaussian-distributed additive noise component with mean $ 0 $ and variance $ \sigma_n^2/2$ in each dimension.  In the case of flat fading channels, the convolution converts to multiplication and the channel response is the Rayleigh-distributed fading gain coefficient with parameter $ \sigma $, i.e.,  $|h_{ab}| \sim \mathrm{Rayleigh} (\sigma)$, and the phase is uniformly distributed, i.e., $\phi(h_{ab}) \sim U [-\pi,\pi] $. The same applies when Bob transmits $ x_{\mathrm{Bob}} (t) $ to Alice, she receives
\begin{align}\label{alicereceivedd}
    y_{\mathrm{Alice}}(t)=x_{\mathrm{Bob}}(t) \circledast \widetilde{h}_{ab} (t) + n_a(t).
\end{align}
The distribution of the channel coefficients $h_{ab}$ will be slightly different in Section\,\ref{nuysimsetupnumericalresults} for the numerical evaluation of the protocol assuming realistic 5G mmWave coefficients. More specifically, samples of Rayleigh distribution are replaced with realistic 5G mmWave channel coefficients considered in \cite{NYUSIM}. 

Wireless channels are essentially reciprocal \cite{mathur2008radio}, meaning that the CSI observed at Bob's end from Alice is the same as Alice's end from Bob assuming an underlying synchronization mechanism. The reciprocity property, i.e., $ h_{ab} \approx \widetilde{h}_{ab}$, is the key to most of the secret key generation protocols that utilize characteristics of the physical layer channel. Also, Alice and Bob are assumed to use OFDM. Suppose that Alice and Bob transmit the $j$-th element of the vectors $ \textbf{x}_{\mathrm{Alice}} (t) $ and $ \textbf{x}_{\mathrm{Bob}} (t) $, respectively, over the $j$-th OFDM subcarrier. The received signals are expressed as follows:
\begin{align}\label{alicebobofdmd}
    \textbf{y}_{\mathrm{Alice}}(t)=\textbf{x}_{\mathrm{Bob}}(t) \circ \widetilde{\textbf{h}}_{ab} (t) + \textbf{n}_a(t),\\
    \textbf{y}_{\mathrm{Bob}}(t)=\textbf{x}_{\mathrm{Alice}}(t) \circ \textbf{h}_{ab} (t) + \textbf{n}_b(t),
\end{align}
where $ \circ $ denotes the Hadamard product, i.e., the element-wise product. By using the received signals at Alice and Bob together with the uniqueness of wireless channel coefficients between them they aim at extracting a shared secret key. Note that in addition to the wireless channel, Alice and Bob are assumed to share a noiseless public channel that Eve has access to. This channel can be realized by using appropriate off-the-shelf modulation and channel coding schemes.

\subsection{Secret Key Generation Using a Relay}\label{sys_relay}

\begin{figure}
    \begin{center}
    \includegraphics[trim=91.37pt 136.06pt 231.56pt 149.14pt,clip, angle=-90,origin=c,width=0.95\columnwidth]{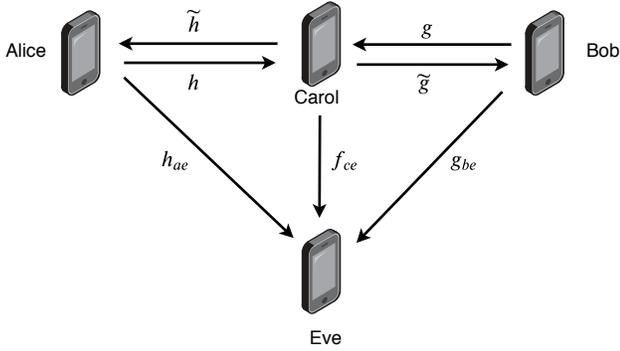}
    \vspace{-0.75in}
    \caption{System model for relay-based secret key generation.}
    \label{fig:RSKG_CH_Model}
    \end{center}
\end{figure}
    
In this case, Alice and Bob do not have access to a direct wireless channel. Instead, there is an intermediate party, also referred to as Carol, operating as a relay node with whom Alice and Bob share authenticated wireless channels which are not secure. The relay is considered to be amplify-and-forward and is assumed to be compliant with the protocol, i.e., it amplifies and forwards the signals without tampering with their contents. However, the relay is considered to be untrusted. This is because it might be susceptible to hacking attacks by an adversary or might be simply curious to learn the contents communicated between Alice and Bob. The eavesdropper Eve is considered to be a passive eavesdropper. The system model is shown in Figure\,\ref{fig:RSKG_CH_Model}. Similar to the model discussed in Section\,\ref{sys_direct}, the channel between each two entities is modeled as a wireless fading channel. Alice, Bob, and Carol utilize OFDM in their transmissions. Alice wishes to transmit a signal $\textbf{x}_{\mathrm{Alice}}(t)$ to Bob through the relay node Carol. First, Alice transmit $\textbf{x}_{\mathrm{Alice}}(t)$ to Carol, who receives
\begin{align}\label{alicerelay}
    \textbf{y}_{\mathrm{Carol}}(t)=\textbf{x}_{\mathrm{Alice}}(t) \circ \textbf{h} (t) + \textbf{n}_r(t).
\end{align}
Next, the relay amplifies the signal using amplification factor $\alpha$ and forwards the amplified signal to Bob. Bob receives
\begin{align}\label{relaybob}
    \textbf{y}_{\mathrm{Bob}}(t)\!&=\!\alpha\!\circ \!\textbf{y}_{\mathrm{Carol}}(t) \! \circ \widetilde{\textbf{g}}(t)+ \textbf{n}_b(t)\\
    &=\!\alpha\!\circ\!(\textbf{x}_{\mathrm{Alice}}(t) \circ \textbf{h} (t) + \textbf{n}_r(t)) \!\circ\!\widetilde{\textbf{g}}(t) + \textbf{n}_b(t),
\end{align}
which holds due to the use of OFDM in transmissions between Alice, Carol, and Bob. The same applies when Bob transmits $\textbf{x}_{\mathrm{Bob}}(t)$ to Alice through Carol. Alice receives
\begin{align}\label{relayalice}
    \textbf{y}_{\mathrm{Alice}}(t)\!=\!\alpha\!\circ\!(\textbf{x}_{\mathrm{Bob}}(t) \circ \textbf{g} (t) + \textbf{n}_r(t)) \!\circ\!\widetilde{\textbf{h}}(t) + \textbf{n}_a(t),
\end{align}
where the $j$-th elements of the vectors $\textbf{h}(t)$, $\widetilde{\textbf{h}}(t)$ are both circularly-symmetric Gaussian-distributed with mean 0 and dimension-variance $\sigma_h^2/2$, and $\textbf{g}(t)$, $\widetilde{\textbf{g}}(t)$ are also circularly-symmetric Gaussian-distributed with mean 0 and dimension-variance $\sigma_g^2/2$. On the other hand, the $j$-th elements of the vectors $\textbf{n}_r(t)$, $\textbf{n}_a(t)$, and $\textbf{n}_b(t)$ are independent and circularly-symmetric Gaussian-distributed with mean 0 and dimension-variance $\sigma_{n_R}^2/2$, $\sigma_{n_A}^2/2$, and $\sigma_{n_B}^2/2$, respectively.

Finally, Alice and Bob use their received signals and utilize the uniqueness of the wireless channel coefficients between them and Carol to extract a secret key. As in the direct secret key generation scenario, a noiseless public channel is available between Alice and Bob through Carol which Eve has access to. Such a channel can be realized by using appropriate off-the-shelf modulation and coding schemes from Alice to Carol, from Carol to Bob, and vice versa.

\noindent
\textbf{Remark:} The relay employs an amplify-and-forward (AF) function with amplification factor $\alpha$. The relay node is placed such that $\alpha$ can be selected according to a certain desired criterion such as maintaining the average transmitted power at the relay, or maintaining the average SNR at the receiver, see, e.g., \cite{hasna2003end} for a detailed discussion. To implement the protocols proposed in this paper, a similar criterion can be adopted since the aim is to create highly correlated sequences which depends on the average received SNR at Alice and Bob. Also, for simplicity it is assumed that the amplification factor is the same for the transmissions to Alice and Bob; however, it can be different for each of them to achieve some specific metric such as the received SNR. The reciprocity property of the indirect channel between Alice and Bob holds in this scenario, since the reciprocity of the individual channels between Alice and Carol, and Carol and Bob still holds.
    
\subsection{Evaluation Metrics for SKG Protocols}
Metrics that are often used to evaluate the performance of secret key generation protocols are as follows \cite{cgc}:
\begin{enumerate}
    \item Bit Generation Rate (BGR): This measures the number of bits per packet in the quantized sequences generated by Alice and Bob, denoted by $\textbf{q}_{a}$ and $\textbf{q}_{b}$, respectively.
    \item Bit Mismatch Rate (BMR): This measures the ratio of the number of bits that are mismatched between $\textbf{q}_{a}$ and $\textbf{q}_{b}$. This quantity can be also measured at Eve's side. Note that the BMR at Eve should be higher than the BMR measured between Alice and Bob; otherwise, no secret key can be generated.  
    \item Bit Error Rate (BER): This measures the ratio of the number of bits that do not match in the final key generated by Alice and Bob as the output of the protocol. This quantity can be also measured at Eve's side, which, ideally, should be close to $50\%$.
    \item Randomness: This indicates whether the final key bit sequence generated by the protocol, denoted by $\textbf{K}_{ab}$, is indistinguishable from a random binary bit sequence. This is often tested using the NIST statistical test suite \cite{rukhin2001statistical}.
\end{enumerate}

In addition to the aforementioned metrics, we introduce a new parameter, referred to as \textit{randomness efficiency}, to measure the length of the shared sequence normalized by the total amount of randomness available to Alice and Bob. Let $R_Q$ denote the total number of shared random bits after quantization. The randomness efficiency, denoted by $E_R$, is defined as
\begin{align}\label{randeff}
    E_R\,\deff\,\frac{R_Q}{H(S)+H(V)},
\end{align}
where $H(S)$ and $H(V)$ are the entropy of Alice's and Bob's sources of randomness, respectively.
    
\section{Proposed Protocols}\label{pp}
The proposed protocols for both scenarios, i.e., secret key generation using a direct channel and relay-based secret key generation, can be partitioned into four stages: induced randomness exchange, quantization, reconciliation, and privacy amplification together with consistency checking. The first stage, i.e., induced randomness exchange, is done differently in the two considered scenarios, while the remaining stages are similar. 

In the first stage, the randomness is induced by Alice and Bob at each of the $N$ OFDM subcarriers, provided that each two-way exchange is done within the same coherence time interval. After the exchange of induced randomness, Alice and Bob process what they receive by performing quantization followed by reconciliation to correct the disparities between their bit sequences. As a result, they obtain, with high probability, identical bit sequences. Then, they use privacy amplification to improve the security of the generated bit sequences. Finally, they check whether their keys are consistent or not. If the keys are not consistent, they re-initiate a new session. The notations for various vectors in the protocol are summarized in Table \ref{table:notation}. Also, Figure \ref{fig:DSKGprotocol} shows an overview of a single session of the key generation protocol for the scenario involving a direct channel, and Figure \ref{fig:RSKGprotocol} shows a single session of the relay-based secret key generation protocol. For ease of notation, we remove the time index $ t $ from the functions while keeping in mind that the exchanges are done within the same coherence time. Next, detailed descriptions of various stages of the proposed protocols are discussed. 

\begin{table}
\begin{center}
\caption{Notation Summary for the $i$-th SKG Session}
\begin{tabular}{c||c}
	\hline
	\textbf{Symbol} & \textbf{Description}\\ 
	\hline
	\textbf{p} & Known probing vector\\
	\hline
	 $\textbf{s}_i$ & Alice's local randomness\\
	\hline
	 $\textbf{v}_i$ & Bob's local randomness\\
	\hline
	 $\textbf{h}_{i,ab}$&Channel coefficients from Alice and Bob \\
	\hline
	 $\widetilde{{\textbf{h}}}_{i,ab}$& Channel coefficients from Bob to Alice \\
	 \hline
	 $\textbf{h}_i$&Channel coefficients between Alice and the relay\\
	\hline
	$\widetilde{\textbf{h}}_i$&Channel coefficients between the relay and Alice\\
	\hline
	$\textbf{g}_i$&Channel coefficients between Bob and the relay\\
	\hline
	$\widetilde{\textbf{g}}_i$&Channel coefficients between the relay and Bob\\
	\hline
	$\textbf{h}_{i,ae}$& Channel coefficients between Alice and Eve\\
	\hline
	$\textbf{h}_{i,be}$& Channel coefficients between Bob and Eve\\
	\hline
	 $\textbf{w}_{i,ab}$ & Alice's samples used for quantization\\
	\hline
	$\widetilde{\textbf{w}}_{i,ab}$ & Bob's samples used for quantization\\
	\hline
	$\textbf{q}_{i,a}$& Alice's quantized version of $\textbf{w}_{i,ab}$\\
	\hline
	$\textbf{q}_{i,b}$& Bob's quantized version of $\widetilde{\textbf{w}}_{i,ab}$\\
	\hline
	$\textbf{K}_{i,ab}$ & Alice's key bits \\
	\hline
	$\widetilde{\textbf{K}}_{i,ab}$& Bob's key bits \\
	\hline
	$\textbf{C}_{i,ab}$ & Alice's check sequence bits \\
	\hline
	$\widetilde{\textbf{C}}_{i,ab}$&Bob's check sequence bits\\
	\hline
\end{tabular}\label{table:notation}
\end{center}
\end{table}

\subsection{Induced Randomness Exchange}
\label{inducedrandomness}
In this stage we aim at creating highly correlated yet random observations at Alice and Bob. We discuss this stage separately for the two considered scenarios as follows:
\subsubsection{Direct Induced Randomness Exchange} \label{p2pskg}
\begin{figure}
\begin{center}
	\includegraphics[trim=214.87pt 335.25pt 127.72pt 77.8pt, clip,width=0.95\columnwidth]{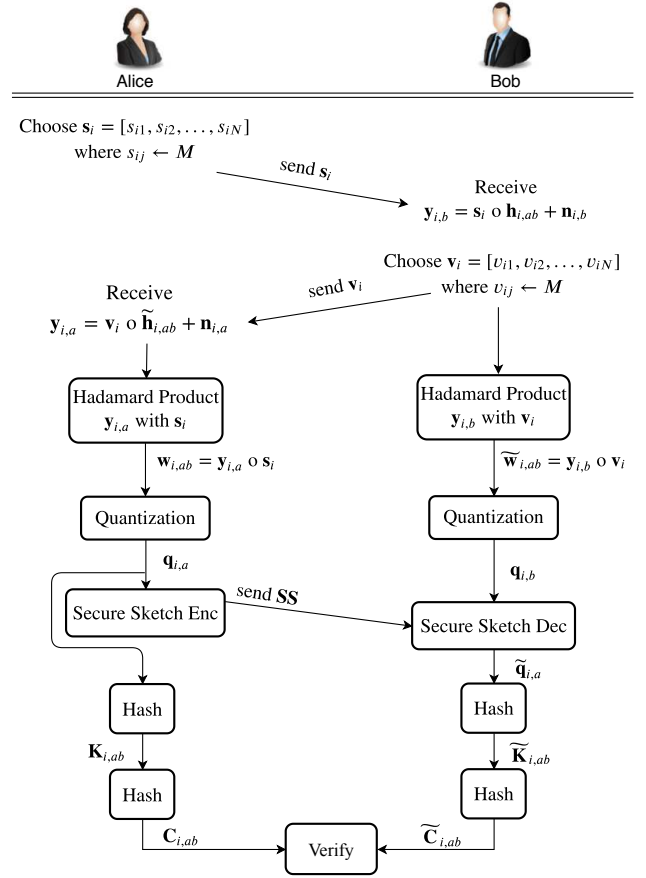}
	\caption{Direct secret key generation protocol overview of a single session.}
	\label{fig:DSKGprotocol}
\end{center}
\end{figure}
In this stage, Alice and Bob exchange randomly generated symbols with each other. In the $i$-th session, Alice chooses a vector $ \textbf{s}_i$ of length $N$ and Bob also chooses a vector $ \textbf{v}_i $ of length $N$. Each element of the vectors $ \textbf{s}_i$ and $ \textbf{v}_i$ is chosen independently and uniformly at random from a set of $ M $ symbols in a $M$-QAM constellation. Then, the symbols are multiplied by a pulse/carrier signal for transmission. The reason behind choosing the symbols from $M$-QAM constellation is that the hardware for transmitting and receiving QAM symbols is readily available in many wireless devices. After the exchange of random symbols, Alice and Bob multiply what they sent with what they received. This results in random sequences $\textbf{w}_{i,ab}$ and $\widetilde{\textbf{w}}_{i,ab}$ available at Alice and Bob, respectively, as follows: 
\begin{align}
    \textbf{w}_{i,ab}=\textbf{s}_i \circ \textbf{v}_i \circ \widetilde{\textbf{h}}_{i,ab}+ \textbf{s}_i \circ \textbf{n}_{i,a}& \label{alicekeydskg},\\
    \widetilde{\textbf{w}}_{i,ab}=\textbf{s}_i \circ \textbf{v}_i \circ \textbf{h}_{i,ab}+ \textbf{v}_i \circ \textbf{n}_{i,b}&.\label{bobkeydskg}
\end{align}
These two vectors are random and highly correlated, as will be shown, which makes them suitable for extracting shared secret keys between Alice and Bob.

\subsubsection{Relay-Based Induced Randomness Exchange}\label{relayskg}
\begin{figure}
\begin{center}
	\includegraphics[trim=26pt 210.36pt 33.31pt 223.1pt,clip, angle=-90,origin=c,width=\columnwidth]{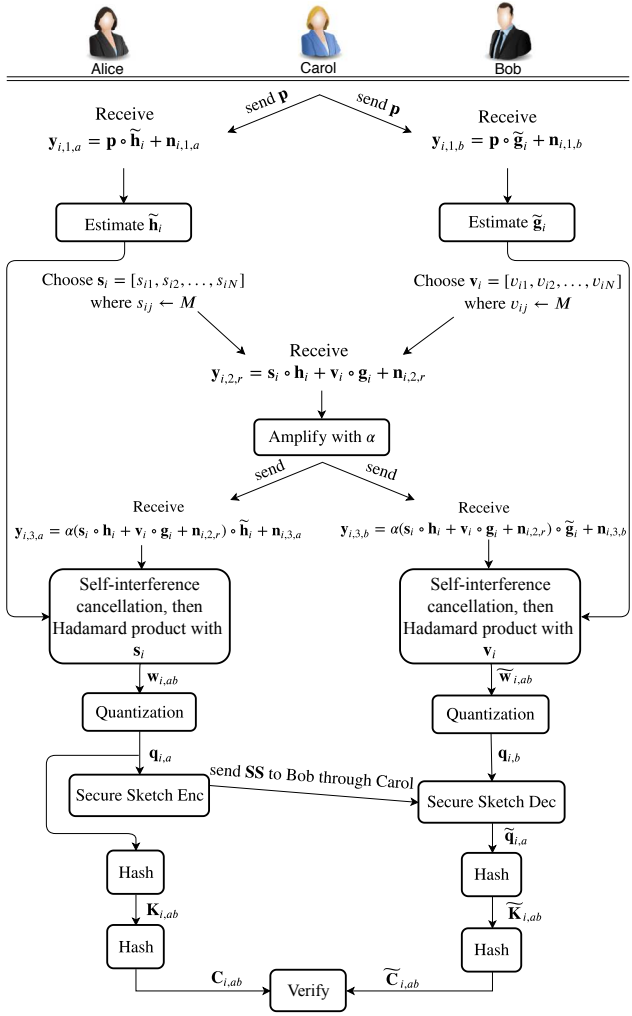}
	\caption{Relay-based secret key generation protocol overview of a single session.}
	\label{fig:RSKGprotocol}
\end{center}
\end{figure}
First, the relay transmits a known probing vector $\textbf{p}$ to Alice and Bob, who receive $\textbf{y}_{i,1,a}$ and $\textbf{y}_{i,1,b}$, respectively, specified as follows:
\begin{align}
    \textbf{y}_{i,1,a}=\textbf{p} \circ \widetilde{\textbf{h}}_i+\textbf{n}_{i,1,a},\\
    \textbf{y}_{i,1,b}=\textbf{p} \circ \widetilde{\textbf{g}}_i+\textbf{n}_{i,1,b}.
\end{align}
Alice and Bob then estimate the channels between themselves and the relay, i.e., $\widetilde{\textbf{h}}_i$ and $\widetilde{\textbf{g}}_i$, respectively, using their observations. Their estimates are denoted by $\widehat{\textbf{h}_i}$ and $\widehat{\textbf{g}}_i$ with estimation errors defined as $\textbf{z}_{i,a}=(\textbf{h}_i \circ \widetilde{\textbf{h}}_i-\widehat{\textbf{h}}_{i}^{\circ 2})$ and $\textbf{z}_{i,b}=(\textbf{g}_i \circ \widetilde{\textbf{g}}_i-\widehat{\textbf{g}}_{i}^{\circ 2})$, respectively, where $(.)^{\circ 2}$ denotes the element-wise square operation. Alice and Bob utilize their respective channel estimates together with their respective local randomness to eliminate the self-interference terms and to generate the correlated samples, to be described next.

Alice and Bob generate, independently and uniformly at random, vectors of length $N$ consisting of $M$-QAM symbols. Let $\textbf{s}_i$ and $\textbf{v}_i$ denote Alice's and Bob's vectors, respectively. They use the probing vector $\textbf{p}$ also for synchronization and, simultaneously, transmit their vectors to the relay in such a way that the received SNRs at the relay with respect to the received sequences from Alice and Bob are the same, and equal to a predetermined value. The relay receives
\begin{align}
	\textbf{y}_{i,2,r}=\textbf{s}_i \circ \textbf{h}_i + \textbf{v}_i \circ \textbf{g}_i+\textbf{n}_{i,2,r}.
	\end{align}
Then, it amplifies $\textbf{y}_{i,2,r}$ with an amplification factor $\alpha$, which is chosen to meet a specific SNR at Alice and Bob, and forwards the amplified signal to Alice and Bob who receive $\textbf{y}_{i,3,a}$ and $\textbf{y}_{i,3,b}$, respectively, as follows:
\begin{align}
	\label{y3a}
	\textbf{y}_{i,3,a}&=\alpha (\textbf{s}_i \circ \textbf{h}_i + \textbf{v}_i \circ \textbf{g}_i+\textbf{n}_{i,2,r}) \circ \widetilde{\textbf{h}}_i +\textbf{n}_{i,3,a},\\
	\label{y3b}
	\textbf{y}_{i,3,b}&=\alpha (\textbf{s}_i \circ \textbf{h}_i + \textbf{v}_i \circ \textbf{g}_i+\textbf{n}_{i,2,r}) \circ \widetilde{\textbf{g}}_i +\textbf{n}_{i,3,b}.
\end{align}
The value of the amplification factor $\alpha$ is assumed to be publicly known. Alice and Bob utilize what they receive from the relay together with their locally generated vectors, their channel estimates, and $\alpha$ in order to construct highly correlated samples. More specifically, the self-interference terms $\alpha\textbf{s}_i \circ \textbf{h}_i  \circ \widetilde{\textbf{h}}_i$ and $\alpha\textbf{v}_i \circ \textbf{g}_i \circ \widetilde{\textbf{g}}_i$ are cancelled at Alice and Bob, respectively, using their local randomness and the channel estimates. The results are normalized by $\alpha$ and then multiplied by the local randomness, which results in $\textbf{w}_{i,ab}$ and $\widetilde{\textbf{w}}_{i,ab}$ at Alice and Bob, respectively, as follows:
\begin{align}
	\textbf{w}_{i,ab}&=\textbf{s}_i \circ  \textbf{v}_i \circ \textbf{g}_i  \circ \widetilde{\textbf{h}}_i +\widehat{\textbf{n}}_{i,3,a} \label{wia}, \\ 
	\widetilde{\textbf{w}}_{i,ab}&=\textbf{s}_i \circ  \textbf{v}_i \circ \widetilde{\textbf{g}}_i  \circ \textbf{h}_i + \widehat{\textbf{n}}_{i,3,b} \label{wib},
\end{align}
where
\begin{align}
	\widehat{\textbf{n}}_{i,3,a}= \textbf{s}^{\circ 2}_i \circ  \textbf{z}_{i,a} +\textbf{s}_i \circ \textbf{n}_{i,2,r}  \circ \widetilde{\textbf{h}}_i  +\textbf{s}_i \circ \textbf{n}_{i,3,a}/\alpha,  \label{noiseA}\\
	\widehat{\textbf{n}}_{i,3,b}=\textbf{v}^{\circ 2}_i \circ  \textbf{z}_{i,b} + \textbf{v}_i \circ \textbf{n}_{i,2,r}  \circ \widetilde{\textbf{g}}_i  +\textbf{v}_i \circ \textbf{n}_{i,3,b}/\alpha,  \label{noiseB}
\end{align}
are the noise terms. The two vectors $\textbf{w}_{i,ab}$ and $\widetilde{\textbf{w}}_{i,ab}$ observed by Alice and Bob are highly correlated and random at each session, which makes them suitable for extracting secret keys.

\subsection{Quantization}
\label{quantization}
In this stage, the complex-valued shared sequences $\textbf{w}_{i,ab}$ and $ \widetilde{\textbf{w}}_{i,ab}$ are turned into bit streams. We use a similar quantization method as suggested in \cite{asbg}. A brief description of the quantization scheme is included next. After collecting the complex-valued measurements $\textbf{w}_{i,ab}$ and $ \widetilde{\textbf{w}}_{i,ab}$, they are sorted as shown in Figure \ref{fig:sorting}. Then, Alice and Bob find the range of sorted data, which is defined as the difference between the maximum value and the minimum value of the sorted vectors. Then, using the range and the quantization resolution $ \delta $, they identify $ \Delta=2^\delta $ uniform quantization intervals, and assign a Gray-code sequence to each interval. Finally, they map each sample to its quantized bit sequence based on the interval it belongs to. The resulting bit sequences for Alice and Bob are denoted by $\textbf{q}_{i,a}$ and $\textbf{q}_{i,b}$, respectively.
	
\begin{figure}
\begin{center}
	\includegraphics[trim=98.47pt 639.52pt 288.2pt 57.14pt, clip,width=0.6\columnwidth]{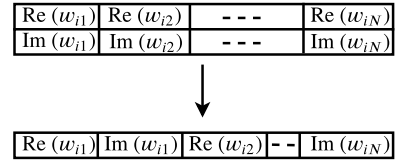}
\end{center}
    \caption{Sorting $ \textbf{w}_{i,ab} $ values before feeding them to the quantizer.}
	\label{fig:sorting}
\end{figure}

\subsection{Reconciliation}

The aim of this stage is to mitigate disagreements between Alice's and Bob's quantized bit sequences. To this end, various methods, such as error-correcting codes, can be used. In our protocols we use error-correcting code-based secure sketch \cite{juels1999fuzzy}, while picking a convolutional code as the underlying code. The reason to pick convolutional codes is due to the simplicity of the encoding process using shift registers and the decoding process using Viterbi decoders \cite{viterbi1967error}. A formal definition of a general secure sketch scheme is as follows:
\begin{definition}\cite{dodis2004fuzzy}\label{securesketch}
    An $(\mathcal{M},m_1,m_2,t)$-secure sketch scheme consists of a sketch function, and a recovery function such that the following properties hold:
    \begin{enumerate}
        \item The sketch function takes an input $w \in \mathcal{M}$ and returns a randomized $SS(w) \in \{0,1\}^*$.
        \item The recovery function takes $SS(w)$ and $\tilde{w} \in \mathcal{M}$, and returns $w$ with probability one as long as the distance between $w$ and $\tilde{w}$ is less than a certain threshold $t$.
        \item For any random variable $W$ over $\mathcal{M}$ with min-entropy $m_1$, an adversary observing $SS(W)$ has an average min-entropy of $W$ conditioned on $SS(W)$ as $\tilde{H}_{\infty}(W |SS(W)) \geq m_2$.
    \end{enumerate}
Note that the min-entropy function of a random variable $X$ is computed as  $H_{\infty}(X)=-\log_2(\underset{x}{\max}( \mathrm{Pr(X=x))})$ and the average min-entropy function of $X$ conditioned on $Y$ is computed as  $\tilde{H}_{\infty}(X|Y) = -\log_2 \big(\underset{y \leftarrow Y}{\mathbb{E}} [2^{-H_{\infty}(X|Y=y)} ] \big) $.
\end{definition}

Next, we describe a construction known as the code-offset secure sketch \cite{juels1999fuzzy}. The encoder is chosen in such a way that the length of its output is equal to the length of $\textbf{q}_{i,a}$. Once the quantized sequences $ \textbf{q}_{i,a} $ and $\textbf{q}_{i,b} $ are available, Alice chooses a bit string $ \textbf{r} $ uniformly at random and encodes it using the convolutional encoder to get $ \mathrm{Enc}(\textbf{r}) $, which is of the same length as $ \textbf{q}_{i,a} $. Then, she computes
\begin{align}
    \textbf{SS}=\textbf{q}_{i,a} \oplus \mathrm{Enc}(\textbf{r}),
\end{align}
where $\oplus$ is the addition modulo 2, and transmits the resulting sequence over the noiseless public channel, either directly as in the first scenario or through the relay as in the second scenario, to Bob. Then, Bob takes the addition modulo 2 of $ \textbf{SS}$ and $\textbf{q}_{i,b}$, feeds it to the Viterbi decoder to get $ \widetilde{\textbf{r}} $, and re-encodes $ \widetilde{\textbf{r}} $ to get $\mathrm{Enc}(\widetilde{\textbf{r}})$. He computes the final sequence as 
\begin{align}
    \widetilde{\textbf{q}}_{i,a}&=\textbf{SS} \oplus \mathrm{Enc}(\mathrm{Dec}(\textbf{SS} \oplus \textbf{q}_{i,b})) \nonumber \\
    &=\textbf{SS} \oplus \mathrm{Enc}( \widetilde{\textbf{r}}).
\end{align}

\noindent\textbf{Remark:} A binary linear code of length $n$ and dimension $k$ with minimum distance $2t+1$ can be used to build an $(\mathcal{F}^n,m_1,m_1-(n-k),t)$-secure sketch scheme, where $\mathcal{F} = \{0,1\}$ for binary codes \cite{dodis2004fuzzy}. The error correction capability of the linear code is related to the underlying rate of the code. This introduces a trade-off between the error correction capability and the security, as higher rates provide better security but can correct less errors, and vice versa. Alice and Bob should start with an initial high rate code and then reduce it accordingly if they observe several consecutive unsuccessful attempts of the protocol.

\subsection{Privacy Amplification Consistency Checking}
Since some information about the shared key is leaked to Eve during the exchange of random symbols and the reconciliation stages, we exploit universal hash functions (UHF) to increase the level of security. In general, UHFs are desired in such scenarios due to their resilience against collisions.
\begin{definition} \cite{hash}\label{uhf}
	A family of hash functions $ H $ that maps a set of inputs $ U $, e.g., binary vectors of length $n$, to a value in the hash table of size $ t $ is called universal if for any two inputs $ x,y  \in U$ with $ x \neq y $, we have
    \begin{align}\label{uhfpr}
	    \underset{h\leftarrow H} {\mathrm{Pr}} (h(x) = h(y) | x \neq y) \leq \frac{1}{t}.
	\end{align}
\end{definition}
    
We also use UHFs to check consistency between keys generated by Alice and Bob, without leaking any information to Eve, as suggested in \cite{keep}. 

Given that $ h $ should be chosen randomly from $ H $, the question is how do we ensure that Alice and Bob agree on the same $h$? We propose a method that guarantees the same choice of $h$ at Alice and Bob if inputs to the UHF are consistent. Suppose we have a random binary sequences $\textbf{q}_{i}$ of length $n$ (This is $\textbf{q}_{i,a}$ for Alice and $\widetilde{\textbf{q}}_{i,a}$ for Bob). For simplicity,  we assume that $ n $ is an even multiple of some integer $m \geq 1$. We divide $\textbf{q}_{i}$ into two sequences of equal length $\textbf{q}_{i}=\textbf{q}_{i,1}\| \textbf{q}_{i,2}$ each of length $n/2$, which is an integer since $n$ is even. Then, $\textbf{q}_{i,1}$ is used to choose $ h $ from $ H $, and $\textbf{q}_{i,2}$ is used as the input to the hash function $h$. Next, a well-known construction of UHF is described that we use in our protocol \cite{hash}. First, the largest prime $ p $ with $2^{m-1}<p<2^m$, i.e., its binary representation consists of $m$ bits, is chosen, where $ m $ is the length of the output bit sequence (such a prime number always exist for $m>1$ by Bertrand's postulate). Then, for $i=1,2$, we divide $\textbf{q}_{i,1}$ and $\textbf{q}_{i,2}$ into $l$ parts $q_{i,1,k}$ and $q_{i,2,k}$ for $k=1,2,\dots,l$, where the length of each part is less than or equal to $m$ bits. For ease of notation, let $q_{i,j,k}$ also denote the number with the binary representation $q_{i,j,k}$. Finally, the following summation is computed:
\begin{align}\label{uhf1}
	h_{\textbf{q}_{i,1}} (\textbf{q}_{i,2}) = \sum_{k=1}^{l} q_{i,1,k} q_{i,2,k} \: \mathrm{mod}  \: p .
\end{align}

Next, the randomness of $\textbf{q}_{i}$ is discussed. In our protocol, $\textbf{s}_i $ and $ \textbf{v}_i $ are chosen uniformly at random for each key generation session. Hence, the value of $\textbf{q}_{i}$ is also random. Therefore, the hash function is randomized during each session, which will be verified in the numerical results section. The output of the aforementioned described hash function is the key bit sequences $\textbf{K}_{i,ab}$ for Alice and $\widetilde{\textbf{K}}_{i,ab}$ for Bob, which are matched with high probability given the reconciliation step. Before Alice and Bob are able to use the key sequences for encryption and decryption, they need to verify the consistency of their keys. To this end, Alice and Bob hash their key sequences $\textbf{K}_{i,ab}$ and $\widetilde{\textbf{K}}_{i,ab}$ again similar to the previously described process. The output of this step is their respective check sequences $\textbf{C}_{i,ab}$ and $\widetilde{\textbf{C}}_{i,ab}$, which they use to verify whether or not their keys are consistent. It is worth noting that, in our protocol, the length of the check sequences, $\textbf{C}_{i,ab}$ and $\widetilde{\textbf{C}}_{i,ab}$, is half the length of the key.

\begin{theorem}\label{pruhff}
	The probability of accepting a mismatched key as consistent by the described protocol with hash table size $p$ for the check sequence is upper bounded as follows:
	\begin{align}\label{beruhf}
	    \mathrm{Pr} (\textbf{C}_{i,ab} = \widetilde{\textbf{C}}_{i,ab} | \textbf{K}_{i,ab} \neq \widetilde{\textbf{K}}_{i,ab}) \leq \frac{1}{p}.
	\end{align}
\end{theorem}

\begin{proof} This follows directly from the definition of universal hash functions, specified in \eqref{uhfpr}, where the output hash table size is $p$.
\end{proof}
	
\section{Attacker Model and the Resilience of Proposed Protocols}\label{attacker}

In this section we discuss eavesdropping strategies by the passive eavesdropper Eve in both scenarios, i.e., whether the communication is through a direct communication channel or through a relay, and provide an upper bound on the probability of a successful eavesdropping attack.

\subsection{Direct Secret Key Generation}\label{attackerDSKG}

In this scenario, Eve's best strategy is to acquire $ \textbf{s}_i, \textbf{v}_i$ and $ \textbf{h}_{i,ab}$. When Alice and Bob exchange signals, Eve receives
\begin{align}
    \textbf{e}_{i,1}&=\textbf{s}_i \circ \textbf{h}_{i,ae} +\textbf{n}_{i,e_1}\label{evereceivedskg1},\\
    \textbf{e}_{i,2}&=\textbf{v}_i \circ \textbf{h}_{i,be} +\textbf{n}_{i,e_2}.\label{evereceivedskg2}
\end{align}
If Eve is able to estimate both $ \textbf{s}_i $ and $ \textbf{v}_i $ from her observations in \eqref{evereceivedskg1} and \eqref{evereceivedskg2} perfectly, she can create samples of the following form:
\begin{align}
    \textbf{w}_{i,e_d,1}&=\textbf{s}_i \circ \textbf{v}_i \circ \textbf{h}_{i,ae} +\textbf{n}_{i,e_3}\label{evekeyseqdskg1},\\
    \textbf{w}_{i,e_d,2}&=\textbf{s}_i \circ \textbf{v}_i \circ \textbf{h}_{i,be} +\textbf{n}_{i,e_4}\label{evekeyseqdskg2}.
\end{align}
Note that she still needs to know $\textbf{h}_{i,ab}$ at all different subcarriers in order to obtain $ \textbf{w}_{i,ab}$ and/or $\widetilde{\textbf{w}}_{i,ab}$, as described in \eqref{alicekeydskg} and \eqref{bobkeydskg}. Luckily, this is, almost, not possible for Eve as discussed next.

In general, the Pearson correlation coefficient $ \rho $ of the channel fading coefficients at locations separated by distance $ d $ can be computed as follows \cite{molisch2012wireless}:
\begin{align}\label{corr}
    \rho= [J_0(kd)]^2,
\end{align}
where  $ J_0(.) $ is the Bessel function of first kind, and $ k $ is the wavenumber. Therefore, if the distance between Alice/Bob and Eve is larger than half of the wavelength $ \lambda/2 $, e.g., 5 cm in 3GHz band, they will experience almost uncorrelated fading channels. Therefore, the leaked information about the generated secret key to Eve is small and is often assumed to be negligible in the literature. However, it is fundamentally important to quantitatively measure the security level. An information-theoretic measure of security is the mutual information between the shared random sequence, from which the secure key will be generated, and what Eve observes. If we assume that the effect of quantization is negligible and also assume that Eve can perfectly recover $ \textbf{s}_i $ and $ \textbf{v}_i $, this mutual information is equal to the mutual information between $\textbf{h}_{i,ab}$ and the pair $(\textbf{h}_{i,ae},\textbf{h}_{i,be})$. One can assume that Eve is closer to Bob than Alice and hence, only consider the mutual information between $\textbf{h}_{i,ab}$ and $\textbf{h}_{i,ae}$ as the dominating term. This can be calculated in each subcarrier as stated in the next lemma. 
\begin{lemma}\label{minfodskg}
    Let $h_{bk}$ and $h_{ek}$ denote the fading coefficients of Bob's and Eve's channels at the $k$-th subcarrier. Also, let $\rho$ denote the correlation coefficient between $h_{bk}$ and $h_{ek}$, specified in \eq{corr}. Then, the mutual information between $h_{bk}$ and $h_{ek}$ is given by 
	\begin{align}
	    I(h_{bk};h_{ek})=-\log(1-\rho^2)  \: \mathrm{bits}.
	\end{align}
\end{lemma}
\begin{proof} 
We have $ h_{bk}=h_{bk,I}+j h_{bk,Q}  $, and $ h_{ek}=h_{ek,I}+j h_{ek,Q}  $. $h_{bk,I}, h_{bk,Q} $ are independent and identically distributed as $ \mathcal{N}(0,\sigma_b^2/2) $ and $h_{ek,I}, h_{ek,Q}$ are independent and identically distributed as $\mathcal{N}(0,\sigma_e^2/2)$. The real parts of Bob's and Eve's channel coefficients are correlated with the parameter $\rho$, and the imaginary parts are also correlated with $ \rho $. Then, we have the following covariance matrices:
\begin{align}\label{sigma1}
    \Sigma_1\!=\!\begin{bmatrix}
       \sigma_b^2/2    & 0 \\
    	0       & \sigma_b^2/2\\
    \end{bmatrix},
    \Sigma_2\!=\!\begin{bmatrix}
    \sigma_e^2/2    & 0 \\
    0       & \sigma_e^2/2\\
    \end{bmatrix},
    \end{align}

    \begin{align}\label{sigma3}
    \Sigma_3\!=\!\begin{bmatrix}
    \sigma_b^2/2    & 0 &\frac{ \rho \sigma_b \sigma_e}{2}&0 \\
    0       & \sigma_b^2/2& 0 &\frac{ \rho \sigma_b \sigma_e}{2}\\
    \frac{ \rho \sigma_b \sigma_e}{2} & 0 &\sigma_e^2/2 &0\\
    0 & \frac{ \rho \sigma_b \sigma_e}{2} & 0 &\sigma_e^2/2\\
    \end{bmatrix}.
\end{align}
Then, the following series of equalities holds:
    \begin{align}\label{dermuinfodskg}
    I(h_{bk};h_{ek})& =I(h_{bk,I}+j h_{bk,Q};h_{ek,I}+j h_{ek,Q} )\nonumber \\
    &\stackrel{\text{(a)}}{=}I(h_{bk,I},h_{bk,Q};h_{ek,I},h_{ek,Q})\nonumber \\
    &\stackrel{\text{(b)}}{=}H_d(h_{bk,I},h_{bk,Q})+ H_d(h_{ek,I},h_{ek,Q}) \nonumber \\
    &\hspace{35pt} -H_d(h_{bk,I},h_{bk,Q},h_{ek,I},h_{ek,Q}) \nonumber \\
    & \stackrel{\text{(c)}}{=} \frac{1}{2} \log(\det(2 \pi e \Sigma_1)) +\frac{1}{2} \log(\det(2 \pi e \Sigma_2)) \nonumber \\
    &\hspace{92pt}  - \frac{1}{2} \log(\det(2 \pi e \Sigma_3))\nonumber \\
    & \stackrel{\text{(d)}}{=} \log(\pi e \sigma_b^2) + \log(\pi e \sigma_e^2)\nonumber \\
    &\hspace{56pt} - \log((\pi e \sigma_b \sigma_e)^2 (1-\rho^2))\nonumber \\
    & \stackrel{\text{(e)}}{=} -\log(1-\rho^2),
    \end{align}
    \noindent where:\\
    (a) holds due to having a one-to-one mapping;\\
    (b) is the expansion of the mutual information expression in terms of differential entropy;\\
    (c) holds by using the well-known expression that the differential entropy of multivariate Gaussian random variables $ \textbf{X}^n\:=(X_1\:,X_2\:,...,\:X_n)$ with covariance matrix $ \Sigma_i $ is $ H_d(\textbf{X}^n) = \frac{1}{2} \log(\det(2 \pi e \Sigma_i)) $;\\
    and (d) and (e) are simplification steps.
\end{proof}

Note that as $ \rho $ goes to zero, the mutual information, given by \Lref{minfodskg}, also goes to zero. 

The next question, which also applies to any physical layer security scheme that utilizes information-theoretic measures of security, is how to quantitatively characterize the chances of a successful eavesdropping attack by Eve, i.e., guessing the key, given the leaked information? The latter is often measured in terms of semantic security, which is a classical notion of security in cryptosystems \cite{goldwasser}. Direct connections between metrics for the information-theoretic security, based on the mutual information, and cryptographic measures of security, including semantic security, are provided in \cite{semantic}. We use these connections to arrive at the following theorem which characterizes the security of the proposed protocol from the aforementioned perspective:

\begin{theorem}\label{semanticdskg}
	Let $N$ denote the number of subcarriers used in the proposed protocol and $\delta$ denote the quantization resolution. Then, the probability of a successful eavesdropping attack by Eve is upper bounded as follows:
    \begin{align}
	    \mathrm{Pr(Successful\: attack)} < &\big(2^{-2\delta}+ \sqrt{2 I(h_{b};h_{e})}\:\big)^{N} +2^{-\delta N},
	\end{align} 
	where $h_{b}$ and $h_{e}$ denote the fading coefficients of Bob's and Eve's channels at a subcarrier.
\end{theorem}
	
\begin{proof} 
\cite[Theorem 5]{semantic} relates the mutual information between Bob's and Eve's observations to the increase in the probability of a successful eavesdropping attack by Eve given her observations. More specifically, the increase in the latter probability is quantified in terms of the mutual information between Bob's and Eve's observations \cite[Theorem 5]{semantic}. Note that the probability that Eve successfully guesses the bits, with no observations, at a single subcarrier is $2^{-2\delta}$. In addition to that, by \cite[Theorem 5]{semantic}, the probability that Eve can guess the shared random bits in a single subcarrier, given her observations in this subcarrier, is increased by at most $\sqrt{2 I(h_{b};h_{e})}$ compared to the case where she does not have any observation. Therefore, Eve's probability of successfully guessing these quantized key bits is upper bounded by $2^{-2\delta}+\sqrt{2 I(h_{b};h_{e})}$. The probability that Eve can recover the shared randomness over all subcarriers is then given by $\big(2^{-2\delta}+\sqrt{2I(h_{b};h_{e})}\big)^N$. Note that $I(h_{b};h_{e})$ is the same across all the subcarriers and is actually computed in terms of $\rho$ in \Lref{minfodskg}. If Eve cannot recover all the shared randomness, the probability that she can guess the secret key correctly, by the property of hash functions in the privacy amplification part of our protocol, is at most $2^{-\delta N}$, when using a key sequence of half the quantized bit sequence length. Utilizing these together with the union bound completes the proof. 
\end{proof}

Note that \Tref{semanticdskg} together with \Lref{minfodskg} can be used to provide a numerical upper bound on the probability of a successful eavesdropping attack given a lower bound on the distance between Eve and both Alice and Bob. For instance, suppose that the distance between Eve and Bob is at least half of a wavelength, i.e., $ d=\lambda/2 $ and is less than the distance between Eve and Alice. Then, the correlation coefficient $\rho$ is at most $ 0.09 $ and by \Lref{minfodskg} the resulting mutual information $I(h_{b};h_{e})$ is at most $0.01$ bits at any of the subcarriers. Suppose that $N=16$ and $\delta=2$, which are also used in the numerical results provided in the next section. Then, by \Tref{semanticdskg}, the probability of a successful attack by Eve given such parameters is at most $2^{-37} + 2^{-32} < 2^{-31}$. 

\subsection{Relay-based Secret Key Generation}
\label{attackerRSKG}
In this scenario, Eve tries to use her observations and the messages transmitted over the public channel to guess $\textbf{w}_{i,ab}$ and/or $\widetilde{\textbf{w}}_{i,ab}$, as described in \eqref{wia} and \eqref{wib}. Her best strategy is to find $\textbf{s}_i$, $\textbf{v}_i$, $\textbf{g}_i$ and $\textbf{h}_i$. When Alice and Bob transmit their induced randomness, Eve receives
\begin{align}\label{evereceiverskg1}
    \textbf{w}_{i,e,1}=\textbf{s}_i \circ \textbf{h}_{i,ae} + \textbf{v}_i \circ \textbf{g}_{i,be} +\textbf{n}_{i,e_5}.
\end{align}
However, when Carol, the relay, amplifies and forwards the signal from Alice and Bob, Eve receives 
\begin{align}\label{evereceiverskg2}
    \textbf{e}_{i,3}=\alpha(\textbf{s}_i \circ \textbf{h}_{i} + \textbf{v}_i \circ \textbf{g}_{i} +\textbf{n}_{i,2,r}) \circ \textbf{f}_{i,ce} + \textbf{n}_{i,e_6}.
\end{align}
Since Eve can estimate the channel coefficients $\textbf{f}_{i,ce}$ from the relay's transmission when it transmits the known probing vector and, also, she knows the value of $\alpha$ from the messages over the public channel, she can successfully estimate 
\begin{align}\label{wie}
        \textbf{w}_{i,e,2}=\textbf{s}_i \circ \textbf{h}_{i} + \textbf{v}_i \circ \textbf{g}_{i}+\textbf{n}_{i,2,r}.
\end{align}
In a worst-case scenario from the legitimate parties' perspective, Eve has as much information as the relay has, in addition to her own observations. Note that this coincides with the problem of securing the shared key against the untrusted relay Carol when Eve is at Carol's location. In the remaining of this section, we analyze the probability of a successful eavesdropping attack assuming that the eavesdropper Eve has all the information available to Carol, in addition to her own observations.

Note that the computations involving the spatial correlation parameter of the wireless channels do not help in ensuring the security in this scenario as they do in the first scenario with a direct communication channel. Also, the mutual information between $\textbf{w}_{i,ab}$, as described in \eqref{wia}, and the pair $(\textbf{w}_{i,e,1},\textbf{w}_{i,e,2})$, as described in \eqref{evereceiverskg1} and \eqref{wie}, respectively, is expected not to be very small as it was in the first scenario. For instance, if this mutual information is greater than $0.5$, then using \cite[Theorem 5]{semantic}, same as in the proof of \Tref{semanticdskg}, does not yield a non-trivial upper bound on the probability of a successful eavesdropping attack. Hence, instead of utilizing semantic security, we need to use an alternative approach to relate $I(\textbf{w}_{i,ab};\textbf{w}_{i,e,1},\textbf{w}_{i,e,2})$ to the probability of a successful eavesdropping attack. To this end, we use Fano's inequality to bound the probability of successful estimation of the quantized bits $\textbf{q}_{i,a}$ by the eavesdropper in terms of the conditional entropy of the quantized bits $\textbf{q}_{i,a}$ given the eavesdropper's observations $(\textbf{w}_{i,e,1},\textbf{w}_{i,e,2})$. Note that the latter can be bounded in terms of $I(\textbf{w}_{i,ab};\textbf{w}_{i,e,1},\textbf{w}_{i,e,2})$. The details of this analysis are given next in the proof of \Tref{fanorskg}. 

To simplify the expressions in the next theorem, let us consider an arbitrary subcarrier and denote the corresponding entries of the vectors $\textbf{w}_{i,ab}$, $\textbf{w}_{i,e,1}$, $\textbf{w}_{i,e,2}$, and $\textbf{q}_{i,a}$ as $w_{ab}$, $w_{e,1}$, $w_{e,2}$, and $q_{a}$, respectively. Note that the result of \Tref{fanorskg} does not depend on the choice of the subcarrier. 
    
\begin{theorem}\label{fanorskg}
    Let $N$ denote the number of subcarriers used in the proposed protocol and $\delta$ denote the quantization resolution. Then, the probability of a successful eavesdropping attack by Eve is upper bounded as follows:
    \begin{align}
	    \mathrm{Pr(Successful\: attack)} <&\Big(1-\frac{H(q_{a})- I_{ab,e}-1}{\log_2(|\mathcal{Q}_{A}|)}\Big)^{N}\nonumber\\
	    &+2^{-\delta N},
	\end{align}
	where $I_{ab,e}$ denotes $I(w_{ab};w_{e,1},w_{e,2})$, $\mathcal{Q}_{A}$ denotes the support of $q_{a}$, and $|\mathcal{Q}_{A}|$ denotes its cardinality.
\end{theorem}
    
\begin{proof}
    Let $\mathrm{C}$ denote the event of correct estimation of $q_{a}$ and $\mathrm{E}$ denote the event of erroneous estimation of $q_{a}$ by the eavesdropper. Then, we have the following 
    \begin{align}
    \mathrm{Pr (C)} &=1-\mathrm{Pr (E)}\\
    &\stackrel{\text{(a)}}{\leq} 1- \frac{H(q_{a}|w_{e,1},w_{e,2})-1}{\log_2(|\mathcal{Q}_{A}|)}\\
    &\stackrel{\text{(b)}}{=}1-\frac{H(q_{a})-I(q_{a};w_{e,1},w_{e,2})-1}{\log_2(|\mathcal{Q}_{A}|)}\\
    &\stackrel{\text{(c)}}{\leq}1-\frac{H(q_{a})-I(w_{ab};w_{e,1},w_{e,2})-1}{\log_2(|\mathcal{Q}_{A}|)}\\
    &\stackrel{\text{(d)}}{=}1-\frac{H(q_{a})-I_{ab,e}-1}{\log_2(|\mathcal{Q}_{A}|)},
    \end{align}
    \noindent where:\\
    (a) holds by Fano's inequality \cite{cover2012elements};\\
    (b) is the expansion of conditional entropy;\\
    (c) follows from the data processing inequality because $q_a$ is a deterministic function of $w_{ab}$ and hence, $(w_{e,1},w_{e,2})$, $w_{ab}$, and $q_a$ form a Markov chain;\\
    (d) is a change of the notation of $I(w_{ab};w_{e,1},w_{e,2})$ to $I_{ab,e}$.
    
    Note that the probability of correctly estimating every bit of $\textbf{q}_{a}$, denoted by $\mathrm{Pr (C_N)}$, is equal to the probability of correctly estimating $q_{a}$ over all the $N$ subcarriers, since the computation of mutual information is the same over all subcarriers. Hence, by using the independence of such events across the $N$ subcarriers, we have
    \begin{align}
    \label{fanorskg1}
        \mathrm{Pr (C_N)} \leq \Big(1-\frac{H(q_{a})-I_{ab,e}-1}{\log_2(|\mathcal{Q}_{A}|)}\Big)^N.
    \end{align}
    If Eve cannot recover all the shared randomness bits in a single session, the probability that she correctly guesses the secret key, by the property of hash functions in the privacy amplification part of our protocol, is at most $2^{-\delta N}$. This, together with \eq{fanorskg1}, and using the union bound complete the proof.
\end{proof}
    
Next, we illustrate how \Tref{fanorskg} can be used in a numerical setup to upper bound the probability of a successful eavesdropping attack by Eve. Suppose that Alice and Bob use 64-QAM constellation points to transmit their induced randomness, the received SNR is $23$\,dB at Alice and Bob in \eqref{y3a}, the quantization parameter $\delta$ is $2$, and the number of subcarriers $N$ is $16$. Eve is located close to Carol, but at least $\lambda/2$ away from her. Given these parameters the mutual information $I_{ab,e}=I(w_{ab};w_{e,1},w_{e,2})$ is numerically estimated as $I_{ab,e} \approx 1.39$ bits, and the entropy of the generated key bits is numerically estimated as $H(q_{a}) \approx 3.86$ bits. Then, by \Tref{fanorskg}, the probability of successful eavesdropping attack is upper bounded, approximately, by $2^{-10.57}$.

\section{Numerical Results}\label{numerical_results}
In this section we consider numerical setups with reasonable parameters and evaluate the proposed protocols for the two described scenarios, i.e., when a direct channel exists, and when a relay is used for the key generation, using the metrics described in Section\,\ref{system_model}. Also, in order to provide numerical results in a realistic environment, a certain number of channel coefficients is generated by the NYUSIM Channel Simulator \cite{NYUSIM}. The simulator is used to generate channel coefficients for realistic {5G mmWave} channels from measurement-based models. A description of the three setups is discussed next, followed by numerical results shown for all the considered setups.

\subsection{Setup}

\subsubsection{Direct Secret Key Generation}\label{directsetupnumericalresults}
In this scenario, it is assumed that Alice and Bob communicate over a direct and reciprocal wireless channel. The constellation size for each subcarrier is $M=16$, i.e., the set of $16$-QAM symbols are used as the set from which local randomness is chosen and transmitted by Alice and Bob. Also, $N=16$ OFDM subcarriers are assumed to be available in the channel between Alice and Bob. The quantization is done with $ \delta = 2 $, i.e., the real and imaginary parts of the received symbol in each subcarrier are quantized into one of the four possibilities as discussed in Section\,\ref{quantization}. Finally, the remaining steps including secure sketch, hashing, and consistency checking are performed as discussed in Section\,\ref{pp}.

\subsubsection{NYUSIM-Based Secret Key Generation}\label{nuysimsetupnumericalresults}
In this scenario, it is assumed that Alice and Bob have a direct reciprocal wireless channel where the coefficients are generated by the NYUSIM Channel Simulator \cite{NYUSIM}. They operate in a non-line-of-sight (NLOS) urban micro-cellular environment at $20^{\circ}$C, the operating frequency is $28$ GHz, and the distance between Alice and Bob and Alice and Eve is $10$ meters. The path contains $1$ meter of foliage, and there is an outdoor-to-indoor low loss. Channel coefficients between Alice and Bob and channel coefficients between Alice and Eve are generated by the NYUSIM Channel Simulator over $N=16$ subcarriers. Alice and Bob choose their induced randomness from the set of $16$-QAM symbols, and the quantization is done with $ \delta = 2 $. The remaining steps follow as in the first scenario.

\subsubsection{Relay-Based Secret Key Generation}\label{relaysetupnumericalresults}
In this scenario, it is assumed that Alice and Bob have direct and reciprocal wireless channels with the relay, which can be perfectly estimated. Also, a scenario is considered for eavesdropping, as discussed in Section\,\ref{attackerRSKG}, where Eve uses the relay's observations. Alice and Bob choose their induced randomness from the set of $64$-QAM symbols, and set their power levels in such a way that the average received SNRs at the relay from both Alice and Bob are equal. Then, the amplification vector $\alpha$ is chosen such that the average SNR, which is considered in the results, of Alice and Bob's correlated observations, \eqref{wia} and \eqref{wib}, respectively, is the same. The remaining parameters and steps are similar to the previous scenarios.

\subsection{Results}

\subsubsection{Bit Generation Rate}
For all the setups described above, Alice and Bob exchange their induced randomness over $N=16$ subcarriers, with quantization resolution $\delta=2$ for the real and imaginary parts separately. Note that $16 \times 2 \times 2 = 64$ bits are generated by Alice and Bob during each session of the protocol. Hence, the bit generation rate (BGR) is $64$ bits/packet. The length of the final secret key is $64/2 = 32$ bits. In order to increase Eve's bit error rate, we assume that four blocks of keys, generated during four separate successful sessions, are added together modulo 2 to obtain one final key of length 32 per each four sessions. Such BGR is considered high compared to protocols designed for static channels setups, which have BGR of $\frac{1}{4}$ to $\frac{1}{2}$ bits/packet as in \cite{wang2011fast}, or $8$ bits/packet as in \cite{huang2013fast}, and it is comparable with protocols designed for dynamic environments, such as \cite{cgc} whose BGR is $60-90$ bit/packet.

\subsubsection{Bit Mismatch Rate and Bit Error Rate}
The bit mismatch rate (BMR) and bit error rate (BER) between Alice and Bob, and Alice and Eve are shown in Figure \ref{fig:BMR} and Figure \ref{fig:BER}, respectively, for the three described setups. For the bit mismatch rate, in the direct and NYUSIM-based SKG setups we compare Alice's and Bob's quantized sequences of \eqref{alicekeydskg} and \eqref{bobkeydskg}, respectively, and Alice's and Eve's quantized sequences of \eqref{alicekeydskg} and \eqref{evekeyseqdskg1}, respectively. Also, for the relay-based SKG setup, we compare Alice's and Bob's quantized sequences of \eqref{wia} and \eqref{wib}, respectively, and Alice's and Eve's quantized sequences of \eqref{wia} and \eqref{wie}, respectively. It is worth noting that as the average SNR increases in the NYUSIM-based SKG setup, Eve's BMR decreases but the rate of decrease slows down. It can be observed that an increase of around $3$\,dB of the average SNR is required in the relay-based SKG setup to achieve a BMR similar to the first two setups. In comparison with other protocols for static environments at $20$\,dB, they have BMR of around $1\%$ as in \cite{wang2011fast}, $4\%$ as in \cite{huang2013fast}, $4\%$ and $13\%$ for the direct and relay-based setups as in \cite{guillaume2015secret}. 

As for the BER, we compare Alice's and Bob's final key sequences and Alice's and Eve's final key sequences. It can be observed that the BER at Bob is extremely low due to the requirement of the consistency checking step in the protocol, which only allows keys whose consistency is verified with high probability to be accepted. Note that the main reason for the average BER at Eve being around $50\%$ is the privacy amplification step of the protocol. In addition to that, the cumulative distribution function (CDF) of the BER at Eve's final key at $20$\,dB average SNR for both the direct and NYUSIM-based SKG setups, and $23$\,dB average SNR for the relay-based SKG setup is shown in Figure \ref{fig:BEREVE}. Note that the curves for all the setups are similar because these curves compare the keys which are the addition modulo 2 of four separate outputs of the hash functions at Alice and Eve. The universal hash function generates a uniformly random output resulting in the similarity of the curves. Also, it is observed that the probability of accepting a mismatched key for the aforementioned average SNRs in the direct and relay-based SKG setups is around $0.0015\%$, and for the NYUSIM-based SKG setup is around $0.00152\%$, which are less than $0.00153\%$ as predicted by \Tref{pruhff}. The aforementioned probability is considered to be very low. In comparison, it is far less than the probability of generating mismatched keys of the protocol proposed for direct SKG in static environments in \cite{fang2017manipulatable}, which is at least $3\%$. In addition to that, as discussed in the security evaluation of the protocol, the probability of acquiring the key perfectly by Eve is, at most, $2^{-31}$ and $2^{-10.57}$, for the direct and relay-based SKG, respectively. In comparison, the protocol proposed for direct SKG in static environments in \cite{fang2017manipulatable} has the probability of acquiring the key by Eve in the range $0.09\%-0.47\%$.

\begin{figure}
    \begin{center}
	\includegraphics[width=\columnwidth]{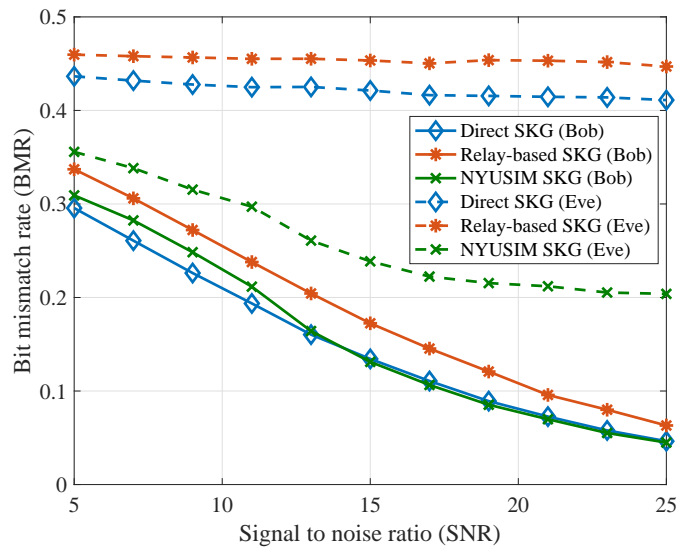}
	\end{center}
	\caption{The bit mismatch rate (BMR) between Alice's sequence and Bob's and Eve's sequences versus the signal to noise ratio.}
	\label{fig:BMR}
\end{figure}

\begin{figure}
    \begin{center}
	\includegraphics[width=\columnwidth]{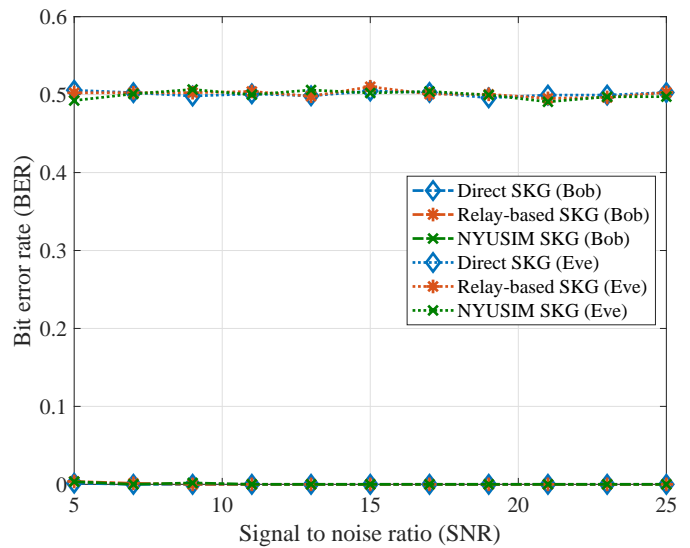}
	\end{center}
	\caption{The bit error rate (BER) between Alice's sequence and Bob's and Eve's sequences versus the signal to noise ratio.}
	\label{fig:BER}
\end{figure}

\begin{figure}
	\begin{center}
		\includegraphics[width=\columnwidth]{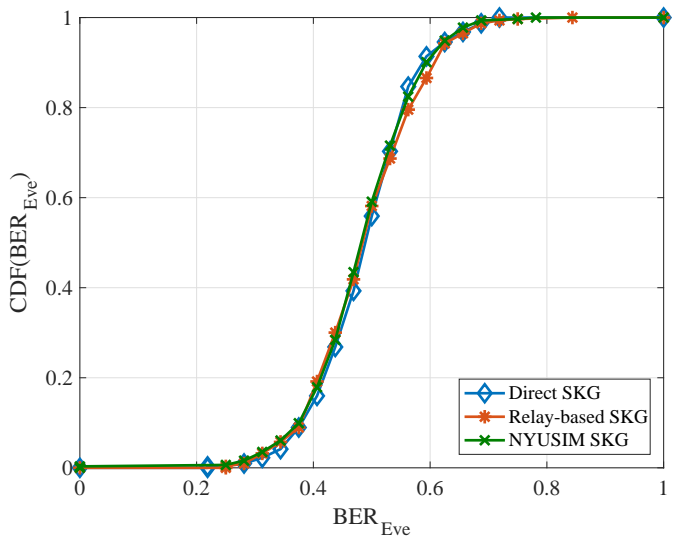}
	\end{center}
	\caption{The cumulative distribution function of the BER at Eve for the direct, relay-based, and NYUSIM-based SKG setups. The compared sequences are the modulo 2 addition of the outputs of four successful key generation sessions. The universal hash function is supposed to generate a uniformly random output, hence the similarity of the curves.}
	\label{fig:BEREVE}
\end{figure}

\subsubsection{Randomness}
The randomness of the generated final key sequence is examined using the NIST statistical test suite \cite{rukhin2001statistical}. The suite consists of 15 tests and generates a probability value, also referred to as \textit{p}-value, for each individual test. For each test, a sequence is considered random with 99\% confidence if the corresponding \textit{p}-value is greater than 0.01. We run the protocol using constant channel coefficients at $20$\,dB average SNRs for the direct and NYUSIM-based SKG setups, and $23$\,dB average SNR for \eqref{wia} and \eqref{wib} in the relay-based SKG setup to generate a sequence of length $ 2^{20} $ bits and feed it to the test suite. Since the sequences pass all the tests as shown in Table \ref{table:nist}, they are considered random with 99\% confidence.
\begin{table}
\begin{center}
	\caption{NIST Statistical Test Results}
	\begin{tabular}{c||c||c||c}
		\hline
		\textbf{Test} & \textbf{Direct} & \textbf{Relay}& \textbf{\!NYUSIM\!}\\ 
		\hline
		Monobit & 0.8712 & 0.9968 & 0.2829  \\ 
		\hline
		Frequency Block & 0.3529  & 0.7458 & 0.3172 \\
		\hline
		Runs & 0.5347 & 0.9781 & 0.4516\\
		\hline
		Longest Run of Ones & 0.7696 & 0.3552 & 0.1692 \\
		\hline
		Binary Matrix Rank & 0.9263 & 0.9974 & 0.3125\\
		\hline
		DFT & 0.6413 & 0.3469 & 0.0121\\
		\hline
		Non-Overlapping Template Matching & 1 & 1 & 1\\
		\hline
		Overlapping Template Matching & 0.2830 & 0.3501 & 0.4043\\
		\hline
		Maurer's Universal Statistical & 0.9991 & 1 & 0.9993\\
		\hline
		Linear Complexity & 0.9909 & 0.5323 &0.0227\\
		\hline
		Serial & 0.2989 & 0.5852 & 0.7236\\
		\hline
		Approximate Entropy & 0.4808  & 0.9160 &0.7529\\
		\hline
		Cumulative Sums & 0.7825 & 0.8392 & 0.1833\\
		\hline
		Random Excursion & 0.0179 & 0.2924 &0.0925 \\
		\hline
		Random Excursion Variant Test & 0.0434 & 0.0154 & 0.0615 \\
		\hline
\end{tabular}\label{table:nist}
\end{center}
\end{table}

\subsubsection{Randomness Efficiency} This is computed according to \eqref{randeff}. For the direct and NYUSIM-based SKG setups, Alice and Bob randomly choose induced randomness bit sequences of length 64, and therefore, $H(S)= H(V) = 64$. Note that the length of the quantized bit sequence is $64$, therefore, $R_Q = 64$. This implies that the randomness efficiency is $ 50\%$. On the other hand, for the relay-based SKG setup, Alice and Bob separately induce $96$ random bits during each round, resulting in $H(S)=H(V)=96$, while the length of the quantized bit sequence is $R_Q=64$. The resulting randomness efficiency of the relay-based SKG setup is $33\%$.  Roughly speaking, the remaining part of the available randomness is used to provide security. The exact trade-off between randomness efficiency and security is an interesting problem.

\subsubsection{Average Number of Sessions Required to Generate Keys} In this part, the average number of sessions Alice and Bob need to generate their final secret key given different values of SNR is compared for all three setups. Note that the length of final secret key is $32$, which is obtained by adding modulo $2$ the outputs of the protocol in four successful sessions. Hence, the average number of sessions required to generate a key approaches $4$ as SNR grows large. The number of required sessions for the relay-based scenario is higher due to a more severe effect of the noise on the shared randomness. This, consequently, affects how often Alice and Bob obtain the same key sequence resulting in a successful session of the protocol. Figure \ref{fig:Sessions} shows the average number of sessions for all considered setups.

\begin{figure}
    \begin{center}
	\includegraphics[width=\columnwidth]{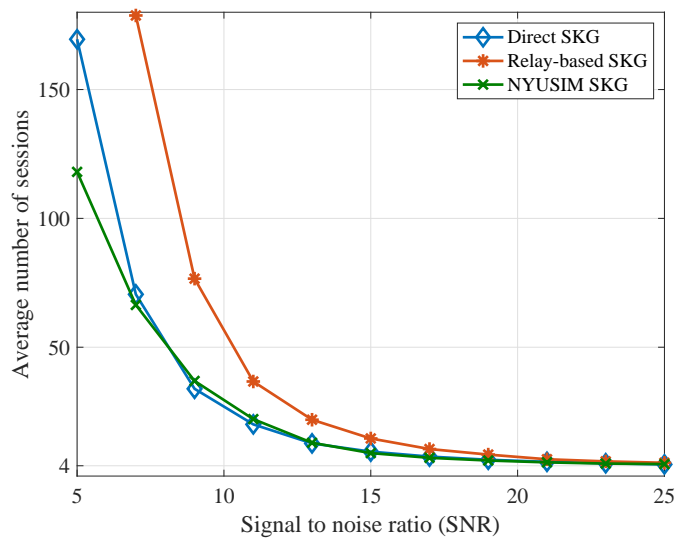}
	\end{center}
	\caption{The average number of sessions until key agreement versus the signal to noise ratio.}
	\label{fig:Sessions}
\end{figure}

\subsubsection{Impact of Non-reciprocity} The perfect channel reciprocity feature is assumed to hold throughout the paper; however, in some practical scenarios, different factors such as mismatched hardware and synchronization errors may cause the channel coefficients experienced at Alice and Bob to not being perfectly reciprocal \cite{cgc,zeng2015physical,primak2014secret}. Such imperfections can be taken into account using the Pearson correlation coefficient, denoted by $\zeta$, between such channel coefficients explained as follows. In general, under perfect channel reciprocity conditions, we have $\zeta=1$, while imperfections reduce the value of $\zeta$. As suggested in \cite{primak2014secret}, a model to describe the relation between the channel coefficients at a subcarrier during session $i$ observed at Alice, i.e., $\widetilde{h}_{i,ab}$, and Bob, i.e., $h_{i,ab}$, when they observe the same SNR is as follows:
\begin{align}
    h_{i,ab} = \zeta  \widetilde{h}_{i,ab}+ \sqrt{1-|\zeta |^2} \frac{\sigma_{h_{i}}}{\sqrt{2}} n_{i},
\end{align}
where $\zeta$ is the correlation coefficient, $\sigma_{h_{i}}^2/2$ is the dimension variance of $h_{i,ab}$ and $\widetilde{h}_{i,ab}$, and $n_{i}$ denotes the circularly-symmetric Gaussian-distributed independent noise component with mean $ 0 $ and unit dimension variance. In order to illustrate the effect of imperfect reciprocity in the direct SKG setup, the bit mismatch rate for different values of the correlation coefficient $\zeta$ is shown in Figure \ref{fig:BMR_nonreciprocity}. It can be observed that as the correlation coefficient between the channel coefficients decreases, the BMR between Alice's and Bob's quantized sequences increases causing the protocol to experience higher number of unsuccessful sessions. For instance, to achieve a BMR around $22\%$, the required SNR is $9$ dB for $\zeta=1$, whereas it is $15$ dB for $\zeta=0.9$. On the other hand, when comparing the average number of sessions required to agree on a key at $15$ dB, it is around $9$ sessions for $\zeta=1$, while it is around $37$ sessions for $\zeta=0.9$. Depending on the severity of the imperfections, the protocol's parameters would require certain adjustments to overcome such degradation. For example, the legitimate parties can decrease the bit generation rate by using a lower quantization resolution $\delta$, or decrease the rate of the error-correcting code used for reconciliation which results in an increase of the amount of information leaked to the eavesdropper.

\begin{figure}
	\begin{center}
		\includegraphics[width=\columnwidth]{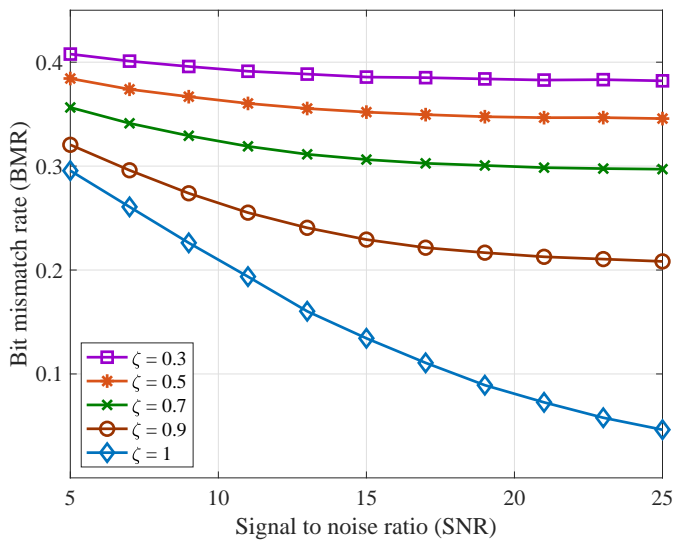}
	\end{center}
	\caption{The bit mismatch rate (BMR) in the direct SKG setup between Alice's and Bob's sequences versus the signal to noise ratio for different values of the correlation coefficient $\zeta$ of the channels experienced at Alice and Bob.}
	\label{fig:BMR_nonreciprocity}
\end{figure}

\section{Conclusion}\label{conclusion}
	
In this paper, we propose a new low-complexity approach to generate secret keys in static environments at high rates using induced randomness. We utilize a low-complexity method where legitimate parties induce locally-generated randomness into the channel such that high-rate common randomness can be generated. More specifically, two main scenarios are considered for the proposed protocols taking into account whether a direct wireless channel is available between legitimate parties or no such channel is available and the transmissions occur through an intermediate relay. We evaluate the reliability and security of the proposed protocols using information theoretic measures. The protocols are also evaluated using metrics including BGR, BMR, BER, and the newly introduced randomness efficiency. Furthermore, numerical results are also shown for a realistic 5G mmWave setup, where channel coefficients are generated by the measurement-based NYUSIM Channel Simulator \cite{NYUSIM}. To ensure that the keys generated by this protocol are random, the generated keys are tested using the NIST statistical test suite. The low-complexity nature of the various steps of the proposed protocols make them appealing for applications concerning resource-constrained devices, e.g., IoT networks, where low complexity methods for generating distributed secret keys are highly desirable. 
	
There are several possible directions for future work. It is interesting to extend the setups considered in this paper to multi-user scenarios where multiple users wish to generate shared secret keys with the help of multiple intermediate relays. From an information-theoretic perspective, this relates to the problem of distributed secret sharing in multi-user scenarios \cite{soleymani2018distributed}. Also, investigating scenarios where the passive eavesdropper has further capabilities than what is considered in this paper, e.g., being able to deploy multiple antennas in the surrounding environment, is another interesting direction. Moreover, studying the resilience of the proposed protocols in the presence of an active eavesdropper who can act as a jammer with the aim of partially crippling the key generation process by sending intentional interference during the randomness exchange is another interesting direction for future work.

\bibliographystyle{IEEEtran}
\bibliography{IEEEabrv}

\end{document}